%% file: On_mKdV_equations_related_to_the_affine_Kac-Moody_algebra_A5_2.tex
\documentclass[11pt,twoside]{gsp}
\usepackage{wrapfig,amssymb,color}

\input somedef
\def\openone{\leavevmode\hbox{\small1\kern-3.3pt\normalsize1}}

\def\lc{\mbox{l.c.\,}}

\def\diag{\mbox{diag\,}}

\def\bbbz{\mathbb{Z}}
\def\bbbd{\mathbb{D}}

\begin{document}

\thispagestyle{plain}

\title{On MKdV Equations Related to the Affine Kac-Moody Algebra $A_{5}^{(2)}$}
\author{V. S. Gerdjikov, D. M. Mladenov, A. A. Stefanov, S. K. Varbev}

\date{}

\maketitle

\author{
Vladimir S. Gerdjikov$^{1}$,
Dimitar M. Mladenov$^{2}$,
Aleksander A. Stefanov$^{2,3,4}$,
Stanislav K. Varbev$^{2,5}$ \\[25pt]
$^1$ Institute of Nuclear Research and Nuclear Energy,\\
Bulgarian Academy of Sciences,\\
72 Tsarigradsko chausee, 1784 Sofia, Bulgaria\\[20pt]
$^2$ Theoretical Physics Department, Faculty of Physics, \\
Sofia University "St. Kliment Ohridski", \\
5 James Bourchier Blvd, 1164 Sofia, Bulgaria\\[20pt]
$^3$ Institute of Mathematics and Informatics,\\
Bulgarian Academy of Sciences,\\
Acad. Georgi Bonchev Str., Block 8, 1113 Sofia, Bulgaria\\[20pt]
$^4$ Faculty of Mathematics,\\
Sofia University "St. Kliment Ohridski",\\
5 James Bourchier Blvd, 1164 Sofia, Bulgaria.\\[20pt]
$^5$ Institute of Solid State Physics\\
Bulgarian Academy of Sciences, \\
72 Tzarigradsko chaussee, 1784 Sofia, Bulgaria.\\
}

\begin{abstract}
We have derived a new system of mKdV-type equations which can be related to the affine Lie algebra $A^{(2)}_{5}$.
This system of partial differential equations is integrable via the inverse scattering method.
It admits a Hamiltonian formulation and the corresponding Hamiltonian is also given.
The Riemann-Hilbert problem for the Lax operator is formulated and its spectral properties are discussed.\\[0.2cm]
\end{abstract}

\label{first}

\section{Introduction}

The general theory of the nonlinear evolution equations (NLEE) allowing Lax representation is well developed
\cite{AKNS,CalDeg,FaTa,GeKu,Ge4,NMPZ}.
In this paper our aim is to derive a set of modified Korteveg--de Vries (mKdV) equations related to three affine Lie algebras using the procedure introduced by Mikhailov \cite{Mikh}.
This means that the equations can be written as the commutativity condition of two ordinary differential operators of the type
\begin{equation}
\begin{split}
\label{1}
L\psi &\equiv \i\frac{\partial \psi}{ \partial x } + U(x,t,\lambda)\psi = 0 \\
M\psi &\equiv \i\frac{\partial \psi}{ \partial t } + V(x,t,\lambda)\psi = \psi \Gamma(\lambda)
\end{split}
\end{equation}
where $ U(x,t,\lambda)$, $V(x,t,\lambda)$ and $\Gamma(\lambda)$ are some polynomials of $\lambda$ to be defined below.
We request also that the Lax pair (\ref{1}) possesses appropriate reduction group \cite{Mikh}, for example
if the reduction group is $\bbbz_{h}$ ($h$ is a positive number) the reduction condition is
\begin{equation}
\label{2}
\begin{aligned}
C\big(U(x,t,\lambda)\big) = U(x,t,\omega\lambda),\qquad C\big(V(x,t,\lambda)\big) = V(x,t,\omega\lambda).
\end{aligned}
\end{equation}
This work can be considered as a continuation of our recent publications \cite{GMSV1,GMSV2,GMSV3}.
Below we consider the three cases separetely.
The underlying Kac-Moody algebras are $B^{(1)}_{2}$, $A^{(2)}_{4}$, $A^{(2)}_{5}$ and the groups of reductions are correspondingly
 $\bbbz_{4}$, $\bbbz_{5} \times \bbbz_{2}$, $\bbbz_{5} \times \bbbz_{2}$.
For the first two cases the Hamiltonians are well known \cite{DrinSok}.
A key motivation for choosing this particular algebras is that the derived equations will have very simple and elegant form.

Section 2 contains a derivation of the mKdV equations related to $B^{(1)}_{2}$. We start with the Lax representation which
is a subject to $\bbbz_{4}$-reduction group \cite{Mikh}, find the equations and derive the corresponding Hamiltonians.
Then using the Lax representation which is a subject to $\bbbz_{5} \times \bbbz_{2}$-reduction group
\cite{Mikh} we derive the system of mKdV equations related to $A^{(2)}_{4}$ and finally derive the corresponding Hamiltonians.
In the next Section 3 we make the same procedure but this time the algebra is $A^{(2)}_{5}$.
The Section 4 is devoted to the spectral properties of Lax operator for each algebra.
Finally we relate this to the famous Riemann-Hilbert problem (RHP).
We finish with  some discussion and  conclusion.

\section{Preliminaries}

\subsection{Equations Related to $B^{(1)}_{2}$}

We assume that the reader is familiar with the theory of semisimple Lie algebras \cite{Helgasson} and affine Lie algebras \cite{Carter}.
The rank of $B^{(1)}_{2}$ is $2$, its Coxeter number is $h=4$ and its exponents are $1, 3$.
Thus the  Coxeter automorphism (see  \cite{GMSV3}) introduces a grading in $B^{(1)}_{2}$ as follows
\begin{equation}
\label{4}
B^{(1)}_{2}=\mathop{\oplus}\limits_{k=0}^{3} \mathfrak{g}^{(k)}.
\end{equation}
The grading condition holds
\begin{equation}
\label{5}
\big[ \mathfrak{g}^{(k)}, \mathfrak{g}^{(l)} \big] \subset \mathfrak{g}^{(k+l)}
\end{equation}
where $k+l$ is taken modulo $4$.

A convenient basis compatible with the grading of $B^{(1)}_{2}$ algebra is \cite{GMSV3}
\begin{equation}\label{223} \begin{aligned}
\mathfrak{g}^{(0)} &: \lc  \{ \mathcal{E}_{11}^{+}, \mathcal{E}_{22}^{+}\}, &\quad \mathfrak{g}^{(1)} &: \lc \{  \mathcal{E}_{12}^{+}, \mathcal{E}_{23}^{+},  \mathcal{E}_{41}^{+} \} \\
\mathfrak{g}^{(2)} &: \lc \{ \mathcal{E}_{13}^{+}, \mathcal{E}_{31}^{+}\}, &\quad \mathfrak{g}^{(3)}&: \lc \{ \mathcal{E}_{21}^{+}, \mathcal{E}_{32}^{+},
\mathcal{E}_{14}^{+}\}
\end{aligned}
\end{equation}
where we use
\begin{equation}
\label{basis}
\mathcal{E}_{ij}^{\pm}=E_{i,j}\mp S_{1}E_{ij}^{T}S_{1}^{-1}=E_{i,j}\mp (-1)^{i+j}E_{6-j,6-i}.
\end{equation}
In this Section $E_{ij}$ is a $5\times 5$ matrix equal to $(E_{ij})_{n,p} =\delta_{in} \delta_{jp}$  and
\begin{equation}
\label{S}
S_{1}=E_{15}-E_{24}+E_{33}-E_{42}+E_{51},  \qquad S_1^2 =\openone
\end{equation}
provides the action of the external automorphism of $A_4 \simeq sl(5)$ related to the symmetry of its Dynkin diagram \cite{Helgasson}.
Obviously all $\mathfrak{E}_{ij}^+$ belong to the subalgebra $B_2\simeq so(5)$ of $A_4$.

For deriving the equations we start with a Lax pair of the form (for details see \cite{GMSV3})
\begin{equation}
\label{7}
\begin{aligned}
L &= \i\partial_x + Q(x,t) - \lambda J\\
M &= \i \partial_t + V^{(0)}(x,t) + \lambda V^{(1)}(x,t) + \lambda^2 V^{(2)}(x,t) - \lambda^{3}K
\end{aligned}
\end{equation}
with
\begin{equation}\label{10}\begin{aligned}
Q(x,t) &=\frac{\i}{2} \left(  u_{1}(x,t)\mathcal{E}_{11}^{+}- u_{2}(x,t)\mathcal{E}_{22}^{+} \right) &\;
J &=\mathcal{E}_{12}^{+}+\mathcal{E}_{23}^{+}+\mathcal{E}_{41}^{+} ,\\
V^{(0)}(x,t) &=v_{1}^{(0)}\mathcal{E}_{11}^{+}+v_{2}^{(0)}\mathcal{E}_{22}^{+}\\
V^{(1)}(x,t) &=v_{1}^{(1)}\mathcal{E}_{12}^{+}+v_{2}^{(1)}\mathcal{E}_{23}^{+}+v_{3}^{(1)}\mathcal{E}_{41}^{+}\\
V^{(2)}(x,t) &=v_{1}^{(2)}\mathcal{E}_{13}^{+}+v_{2}^{(2)}\mathcal{E}_{31}^{+} &\;
K&=2^{6}(\mathcal{E}_{21}^{+}+2\mathcal{E}_{32}^{+}+\mathcal{E}_{14}^{+}) .
\end{aligned}\end{equation}
We require that $[L,M]=0$ for any $\lambda$.
The condition $[L,M]=0$ leads to a set of recurrent relations (see \cite{AKNS,VG-13,VG-Ya-13}) which allow us to determine $V^{(k)}(x,t)$ in terms of the potential $Q(x,t)$ and its $x$-derivatives.

After the  transformation $x\mapsto 2x$ and $t\mapsto 2t$ the equations become
\begin{equation}
\begin{aligned}
&\frac{\partial u_{1}}{\partial t} = 4 \frac{\partial}{\partial x} \biggl( -u^{3}_{1}+3u_{2}^{2}u_{1}-4\frac{\partial^{2} u_{1}}{\partial x^{2}}+6u_{1}\frac{\partial u_{2}}{\partial x} \biggr) \\
&\frac{\partial u_{2}}{\partial t} = 4 \frac{\partial}{\partial x} \biggl( -u^{3}_{2}+3u_{1}^{2}u_{2}+2\frac{\partial^{2} u_{2}}{\partial x^{2}}-6u_{1}\frac{\partial u_{1}}{\partial x} \biggr).
\end{aligned}
\end{equation}
They can be written as  Hamiltonian equations of motion
\begin{equation}
\label{16}
\frac{\partial q_{i}}{\partial t}=\frac{\partial}{\partial x}\frac{\delta H}{\delta q_{i}}
\end{equation}
with the Hamiltonian
\begin{equation}\label{17a}
H=-\int_{-\infty}^{\infty} \d x \biggl( u_{1}^{4}+u_{2}^{4}-6u_{1}^{2}u_{2}^{2}-8 \left( \frac{\partial u_{1}}{\partial x} \right)^{2}
+4 \left( \frac{\partial u_{2}}{\partial x} \right)^{2} -12u_{1}^{2} \left( \frac{\partial u_{2}}{\partial x} \right) \biggr)
\end{equation}
which coincides with the one in  \cite{DrinSok}.

\subsection{Equations Related to $A^{(2)}_{4}$}

Similarly we treat the $A^{(2)}_{4}$ case. The rank of this algebra is $2$, its Coxeter number is $h=10$
and its exponents are $1, 3, 7, 9$. Now the Coxeter automorphism is  of order  $10$
and introduces a grading in  $A^{(2)}_{4}$ as follows
\begin{equation}
\label{20}
A^{(2)}_{4}=\mathop{\oplus}\limits_{k=0}^{9} \mathfrak{g}^{(k)}.
\end{equation}
The grading condition holds
\begin{equation}
\label{21}
\big[ \mathfrak{g}^{(k)}, \mathfrak{g}^{(l)} \big] \subset \mathfrak{g}^{(k+l)}
\end{equation}
where now $k+l$ is taken mod$10$, \cite{DrinSok,Carter}, see also \cite{GMSV3}.

A convenient basis compatible with the grading of $B^{(1)}_{2}$ algebra is \cite{GMSV3}
\begin{equation}\label{eq:223'}\begin{aligned}
\mathfrak{g}^{(0)} &: \lc \{ \mathcal{E}_{11}^{+}, \mathcal{E}_{22}^{+} \}&\qquad
\mathfrak{g}^{(1)}&:  \lc \{ \mathcal{E}_{14}^{-}, \mathcal{E}_{31}^{-}, \mathcal{E}_{42}^{-}\} \\
\mathfrak{g}^{(2)}&: \lc \{ \mathcal{E}_{12}^{+}, \mathcal{E}_{23}^{+}\}&\qquad
\mathfrak{g}^{(3)}&: \lc \{ \mathcal{E}_{21}^{-}, \mathcal{E}_{32}^{-}, \mathcal{E}_{15}^{-}\}&\qquad \\
\mathfrak{g}^{(4)}&: \lc \{ \mathcal{E}_{13}^{+}, \mathcal{E}_{41}^{+}\} &\qquad
\mathfrak{g}^{(5)}&: \lc \{ \mathcal{E}_{11}^{-}-\mathcal{E}_{22}^{-}, \mathcal{E}_{22}^{-}-\mathcal{E}_{33}^{-}\}\\
\mathfrak{g}^{(6)}&: \lc \{ \mathcal{E}_{14}^{+}, \mathcal{E}_{31}^{+}\}&\qquad
\mathfrak{g}^{(7)}&: \lc \{ \mathcal{E}_{12}^{-}, \mathcal{E}_{23}^{-}, \mathcal{E}_{51}^{-}\}\\
\mathfrak{g}^{(8)}&: \lc \{ \mathcal{E}_{21}^{+},  \mathcal{E}_{32}^{+}\}&\qquad
\mathfrak{g}^{(9)}&: \lc \{ \mathcal{E}_{13}^{-}, \mathcal{E}_{41}^{-}, \mathcal{E}_{24}^{-}\}
\end{aligned}\end{equation}
where we have used the basis (\ref{basis}) generated by the same matrix $S_1$ (\ref{S}).

The relevant Lax pair is of the form (for details see \cite{GMSV3})
\begin{equation}
\label{24}
\begin{aligned}
L &= \i\partial_x + Q(x,t) - \lambda J\\
M &= \i \partial_t + V^{(0)}(x,t) + \lambda V^{(1)}(x,t) + \lambda^2 V^{(2)}(x,t) - \lambda^{3}K
\end{aligned}
\end{equation}
with
\begin{equation}\label{27}\begin{aligned}
Q(x,t) &= \i u_{2}(x,t)\mathcal{E}_{11}^{+} - \i u_{1}(x,t)\mathcal{E}_{22}^{+} &\;
J&=\mathcal{E}_{14}^{-}+\mathcal{E}_{31}^{-}+\mathcal{E}_{42}^{-}, \\
V^{(0)}(x,t) &= v_{1}^{(0)}\mathcal{E}_{11}^{+}+(v_{2}^{(0)}-v_{1}^{(0)})\mathcal{E}_{22}^{+} \\
V^{(1)}(x,t) &= v_{1}^{(1)}\mathcal{E}_{14}^{-}+v_{2}^{(1)}\mathcal{E}_{31}^{-}+v_{3}^{(1)}\mathcal{E}_{42}^{-}\\
V^{(2)}(x,t) &=v_{1}^{(2)}\mathcal{E}_{12}^{+}+v_{2}^{(2)}\mathcal{E}_{23}^{+} &\;
K &= 20(\mathcal{E}_{21}^{-}-2\mathcal{E}_{32}^{-}+\mathcal{E}_{15}^{-}).
\end{aligned}\end{equation}
We continue analogously.
The condition $[L,M]=0$ leads to a set of recurrent relations (see \cite{AKNS,VG-13,VG-Ya-13})
which allow us to determine $V^{(k)}(x,t)$ in terms of the potential $Q(x,t)$ and its $x$-derivatives.

After the transformation $x \mapsto 2x$ the equations are
\begin{equation}
\begin{aligned}
\frac{\partial u_{1}}{\partial t} &= -2\frac{\partial}{\partial x} \biggl ( 3\frac{\partial^{2} u_{2}}{\partial x^{2}}
+\frac{\partial^{2} u_{1}}{\partial x^{2}}+ \left( 3u_{2}+6u_{1} \right)\frac{\partial u_{2}}{\partial x}+3u_{2}^{2}u_{1}-2u_{1}^{3} \biggr) \\
\frac{\partial u_{2}}{\partial t} &= -2\frac{\partial}{\partial x} \biggl ( 4\frac{\partial^{2} u_{2}}{\partial x^{2}}
+3\frac{\partial^{2} u_{1}}{\partial x^{2}}- \left( 6u_{1}+3u_{2} \right)\frac{\partial u_{1}}{\partial x}+3u_{1}^{2}u_{2}-2u^{3}_{2} \biggr) .
\end{aligned}
\end{equation}
The Hamiltonian formulation follows from (\ref{16}) and we find for $H$ \cite{DrinSok}
\begin{equation}
\begin{aligned}
H&= \int_{-\infty}^{\infty} dx \biggl( u_{1}^{4}+u_{2}^{4}-3u_{1}^{2}u_{2}^{2}+4 \left( \frac{\partial u_{2}}{\partial x} \right)^{2}+ \left( \frac{\partial u_{1}}{\partial x} \right)^{2}\\
&+3u_{2}^{2} \left( \frac{\partial u_{1}}{\partial x} \right) -6u_{1}^{2} \left( \frac{\partial u_{2}}{\partial x} \right)
+6 \left( \frac{\partial u_{1}}{\partial x} \right) \left( \frac{\partial u_{2}}{\partial x} \right) \biggr) .
\end{aligned}
\end{equation}

\section{Derivation of the Equations Related to $A^{(2)}_{5}$}

Now we consider the twisted affine Kac-Moody algebra $A^{(2)}_{5}$ case.
Its rank is $3$, the Coxeter number is $h=10$
and its exponents are $1, 3, 5, 7, 9$, see \cite{DrinSok,Carter}.
Then the Coxeter automorphism is given by
\begin{equation}\label{eq:C2}
C(X) = C_2 V(X)C_2^{-1}
\end{equation}
where $V$ is the external automorphism of the algebra $A_{5}\simeq sl(6)$ generated by the symmetry of its Dynkin
diagram and $C_2$ is an element of the Cartan subgroup defined below.
More precisely
\begin{equation}\label{eq:V2}
\begin{aligned}
V(X) = - S_2 X^T S_2^{-1}, \quad
S_2 = E_{1,6} -E_{2,5} +E_{3,4} -E_{4,3} + E_{5,2} -E_{6,1}.
\end{aligned}
\end{equation}
Note that in this Section the matrices $E_{kj}$ are $6\times 6$ matrices equal to $(E_{k,j})_{np}=\delta_{kn} \delta_{jp}$;
besides $S_2^2=-\openone$.

In analogy with the previous Section we introduce:
\begin{equation}
\mathcal{E}_{ij}^{\pm}=E_{i,j}\mp S_{2}E_{ij}^{T}S_{2}^{-1}=E_{i,j}\mp (-1)^{i+j}E_{7-j,7-i}
\end{equation}
which obviously satisfy:
\begin{equation}\label{eq:}
\begin{split}
V(\mathcal{E}_{ij}^{+}) = \mathcal{E}_{ij}^{+}, \qquad V(\mathcal{E}_{ij}^{-}) = -\mathcal{E}_{ij}^{-}.
\end{split}
\end{equation}
It is easy to check that $\mathcal{E}_{ij}^{+}$ provide a basis for the subalgebra $sp(6)$ of $A_{5}^{(2)}$.
The Cartan subgroup element $C_2$ is defined by
\begin{equation}\label{eq:C2'}
\begin{split}
 C_{2} = \diag (1, \omega, \omega^{2} ,\omega^{3}, \omega^{4}, 1).
\end{split}
\end{equation}

The basis is as follows
\begin{equation}\label{eq:223''}\begin{aligned}
\mathfrak{g}^{(0)} &: \lc \{ \mathcal{E}_{11}^{+}, \mathcal{E}_{22}^{+}, \mathcal{E}_{33}^{+}\}&\quad
\mathfrak{g}^{(1)} &: \lc \{ \mathcal{E}_{21}^{+}, \mathcal{E}_{32}^{+}, \mathcal{E}_{43}^{+}, \mathcal{E}_{15}^{-}\} \\
\mathfrak{g}^{(2)} &: \lc \{  \mathcal{E}_{31}^{+}, \mathcal{E}_{42}^{+}, \mathcal{E}_{14}^{-}\}&\quad
\mathfrak{g}^{(3)} &: \lc \{\mathcal{E}_{14}^{+}, \mathcal{E}_{25}^{+}, \mathcal{E}_{13}^{-}, \mathcal{E}_{24}^{-}\} \\
\mathfrak{g}^{(4)} &: \lc \{ \mathcal{E}_{51}^{+}, \mathcal{E}_{12}^{-}, \mathcal{E}_{23}^{-}\}&\quad
\mathfrak{g}^{(5)} &: \lc \{ \mathcal{E}_{16}^{+}, \mathcal{E}_{61}^{+}, \mathcal{E}_{33}^{-}-\mathcal{E}_{11}^{-}, \mathcal{E}_{33}^{-}-\mathcal{E}_{22}^{-}\}\\
\mathfrak{g}^{(6)} &: \lc \{ \mathcal{E}_{15}^{+}, \mathcal{E}_{21}^{-}, \mathcal{E}_{32}^{-}\}&\quad
\mathfrak{g}^{(7)} &: \lc \{ \mathcal{E}_{41}^{+}, \mathcal{E}_{52}^{+}, \mathcal{E}_{31}^{-}, \mathcal{E}_{42}^{-}\}\\
\mathfrak{g}^{(8)} &: \lc \{ \mathcal{E}_{13}^{+}, \mathcal{E}_{24}^{+}, \mathcal{E}_{41}^{-}\}&\quad
\mathfrak{g}^{(9)} &: \lc \{ \mathcal{E}_{12}^{+}, \mathcal{E}_{23}^{+}, \mathcal{E}_{34}^{+}, \mathcal{E}_{51}^{-}\}.
\end{aligned}\end{equation}

The grading condition is like always
\begin{equation}
\big[ \mathfrak{g}^{(k)}, \mathfrak{g}^{(l)} \big] \subset \mathfrak{g}^{(k+l)}
\end{equation}
where $k+l$ is taken modulo $10$.

We take a Lax pair of the form
\begin{equation}
\label{24'}
\begin{aligned}
L &= \i \partial_x + Q(x,t) - \lambda J \\
M &= \i \partial_t + V^{(0)}(x,t) + \lambda V^{(1)}(x,t) + \lambda^2 V^{(2)}(x,t) - \lambda^{3}K
\end{aligned}
\end{equation}
where
\begin{equation}
\label{25}
Q(x,t) \in \mathfrak{g}^{(0)}, \quad V^{(k)}(x,t) \in \mathfrak{g}^{(k)}, \quad K \in \mathfrak{g}^{(3)}, \quad J\in \mathfrak{g}^{(1)}.
\end{equation}
This means
\begin{equation}\label{eq:46}\begin{aligned}
Q(x,t) &= \i\sum_{j=1}^{3}q_{j}(x,t)\mathcal{E}_{jj}^{+}, \qquad  \qquad
J =\mathcal{E}_{21}^{+}+\mathcal{E}_{32}^{+}+\frac{1}{2} \mathcal{E}_{43}^{+} +\frac{1}{2}\mathcal{E}_{15}^{-} \\
V^{(0)}(x,t) &=\sum_{j=1}^{3}v_{j}^{(0)}\mathcal{E}_{jj}^{+}\\
V^{(1)}(x,t) &=v_{1}^{(1)}\mathcal{E}_{21}^{+}+v_{2}^{(1)}\mathcal{E}_{32}^{+} +\frac{1}{2} v_{3}^{(1)}\mathcal{E}_{43}^{+}+\frac{1}{2}v_{4}^{(1)}\mathcal{E}_{15}^{-} \\
V^{(2)}(x,t) &=-v_{1}^{(2)}\mathcal{E}_{31}^{+}-v_{1}^{(2)}\mathcal{E}_{42}^{+}-\frac{1}{2}v_{3}^{(2)}\mathcal{E}_{14}^{-}  \qquad \qquad K=b J^{3}.
 \end{aligned}\end{equation}

The condition $[L,M]=0$ leads to a set of recurrent relations (see \cite{AKNS,VG-13,VG-Ya-13}) which allow us to determine
$V^{(k)}(x,t)$ in terms of the potential $Q(x,t)$ and its $x$-derivatives.
For $V^{(2)}(x,t)$ we find, skipping the details, the result
\begin{equation}\label{29} \begin{aligned}
v_{1}^{(2)} = -\i b(q_{1}+q_{2}+q_{3})  \qquad  v_{2}^{(2)} &= -\i bq_{2} \qquad
v_{3}^{(2)} = -\i b(q_{1}-q_{2}-q_{3}).
\end{aligned}
\end{equation}
For $V^{(1)}(x,t)$ we find
\begin{equation}
\begin{aligned}
\label{28}
v_{1}^{(1)} &=2b \left( q_{1}q_{2}-\frac {\partial q_{1}}{\partial x} \right)+f \\
v_{2}^{(1)} &=b \left ( q_{3}^{2}-q_{1}^{2}+q_{1}q_{2}+q_{2}q_{3}+\frac {\partial }{\partial x} \left( q_{3}+q_{2}-q_{1} \right) \right)+f \\
v_{3}^{(1)} &=b \left ( q_{3}^{2}-q_{2}^{2}-q_{1}^{2}+q_{1}q_{2}+\frac {\partial }{\partial x} \left( q_{3}+2q_{2}-q_{1} \right) \right)+f \\
v_{4}^{(1)}&=f
\end{aligned}
\end{equation}
where $f(x,t)$ is some arbitrary function. Using a well known technique from the theory of recursion operators
 \cite{AKNS,Ge4,VG-Ya-13} we find from the equations for $V^{(0)}(x,t)$ also $f(x,t)$
\begin{equation}\label{30}
f = \frac{b}{5} \left( 2q_{2}^{2}+2q_{1}^{2}-3q_{3}^{2}-5q_{1}q_{2} +\frac {\partial }{\partial x} (5q_{1}-4q_{2}-3q_{3}) \right)
\end{equation}
and
\begin{equation}\label{31}\begin{aligned}
v_{1}^{(0)} &= \frac{\i b}{5} \biggl( -5\frac {\partial^{2}q_{1}}{\partial {x}^{2}}+3q_{1}\frac{\partial}{\partial x}(3q_{2}
+q_{3})-2q_{1}^{3}+3q_{1}(q_{2}^{2}+q_{3}^{2}) \biggr) \\
v_{2}^{(0)} &= \frac{\i b}{5} \biggl(  \frac {\partial^{2}}{\partial {x}^{2}}(4q_{2}+3q_{3})+3q_{2}\frac{\partial q_{3}}{\partial x}
-9q_{1}\frac{\partial q_{1}}{\partial x}+6q_{3}\frac{\partial q_{3}}{\partial x}-2q_{2}^{3}+3q_{2}(q_{1}^{2}+q_{3}^{2}) \biggr)\\
v_{3}^{(0)} &= \frac{\i b}{5} \biggl( \frac {\partial^{2}}{\partial {x}^{2}}(q_{3}+3q_{2})-6q_{3}\frac{\partial q_{2}}{\partial x}
-3q_{1}\frac{\partial q_{1}}{\partial x}-3q_{2}\frac{\partial q_{2}}{\partial x}-2q_{3}^{3}+3q_{3}(q_{1}^{2}+q_{2}^{2}) \biggr).
\end{aligned}
\end{equation}
And finally, the $\lambda$-independent terms in the Lax representation provide the equations
\begin{equation}
\begin{aligned}
\label{32}
&\alpha\frac{\partial q_{1}}{\partial t} = \frac{\partial}{\partial x} \biggl( -5\frac {\partial^{2}q_{1}}{\partial {x}^{2}}
+3q_{1}\frac{\partial}{\partial x}(3q_{2}+q_{3})-2q_{1}^{3}+3q_{1}(q_{2}^{2}+q_{3}^{2} \biggr)\\
&\alpha\frac{\partial q_{2}}{\partial t} =  \frac{\partial}{\partial x} \biggl( \frac {\partial^{2}}{\partial {x}^{2}}(4q_{2}
+3q_{3})+3q_{2}\frac{\partial q_{3}}{\partial x}-9q_{1}\frac{\partial q_{1}}{\partial x}+6q_{3}\frac{\partial q_{3}}{\partial x}-2q_{2}^{3}+3q_{2}(q_{1}^{2}+q_{3}^{2}) \biggr)\\
&\alpha\frac{\partial q_{3}}{\partial t} = \frac{\partial}{\partial x} \biggl( \frac {\partial^{2}}{\partial {x}^{2}}(q_{3}
+3q_{2})-6q_{3}\frac{\partial q_{2}}{\partial x}-3q_{1}\frac{\partial q_{1}}{\partial x}-3q_{2}\frac{\partial q_{2}}{\partial x}-2q_{3}^{3}+3q_{3}(q_{1}^{2}+q_{2}^{2}) \biggr)
\end{aligned}
\end{equation}
where $\alpha=\frac{5}{b}$.

We find for the corresponding Hamiltonian (\ref{16})
\begin{equation}
\begin{aligned}
H&=\frac{1}{\alpha}\int_{-\infty}^{\infty} \d x \biggl( -\frac{1}{2}\sum_{i=1}^{3}q_{i}^{4}+\frac{3}{2} \underset{i<j}{\sum_{i=1}^{3}
\sum_{j=1}^{3}}q_{i}^{2}q_{j}^{2}+\frac{5}{2} \left( \frac{\partial q_{1}}{\partial x} \right)^{2}
-2 \left( \frac{\partial q_{2}}{\partial x} \right)^{2}-\frac{1}{2} \left( \frac{\partial q_{3}}{\partial x} \right)^{2}\\
&+\frac{\partial q_{2}}{\partial x} \left( \frac{9}{2}q_{1}^{2}-3q_{3}^{2} \right) +\frac{3}{2}\frac{\partial q_{3}}{\partial x}(q_{1}^{2}
+q_{2}^{2})-3 \left( \frac{\partial q_{2}}{\partial x} \right) \left( \frac{\partial q_{3}}{\partial x} \right) \biggr) .
\end{aligned}
\end{equation}

\section{On the Spectral Properties of the Lax Operators}

\subsection{General Theory}

Here we will outline the general approach of constructing the fundamental analytic solutions (FAS) of the Lax
operators $L$ with deep reductions \cite{Beals-Coifman,GeYa*94,ContM,VG-Ya-13,VYa}. Next we will
detail these results for the three different Lax operators considered above.

Our first remark is about the  fact, that after a simple similarity transformation, which diagonalizes the relevant
matrix $J$, each of the above Lax operators will take the form
\begin{equation}\label{eq:Lt}
\begin{split}
\tilde{L} \equiv \i \frac{\partial \tilde{\chi}}{ \partial x } + (\tilde{Q}(x,t) - \lambda \tilde{J}) \tilde{\chi}(x,t,\lambda)= 0
\end{split}
\end{equation}
where $\tilde{J}$ is a diagonal matrix with complex eigenvalues.

The main ingredient needed for solving the direct and the inverse scattering problem of $\tilde{L}$ are the Jost solutions.

It is well known that the Lax operators of the form (\ref{eq:Lt}) with generic complex-valued $J$ allow Jost solutions only
for potentials on compact support \cite{Beals-Coifman}. An important theorem proved by Beals and Coifman \cite{Beals-Coifman}
states that any smooth potential $\tilde{Q}(x,t)$ can be approximated with an arbitrary precision by a potential on finite support.
Then one can introduce the Jost solutions by
\begin{equation}
\label{eq:Jo}
\lim_{x\to -\infty}\tilde{\phi}_-(x,t,\lambda) \e^{\i J\lambda x} =\openone, \qquad
\lim_{x\to \infty}\tilde{\phi}_+(x,t,\lambda) \e^{\i J\lambda x} =\openone.
\end{equation}
Then the scattering matrix is introduced by
\begin{equation}
\label{eq:T}\begin{split}
T(\lambda,t) = \widehat{\tilde{ \phi}}_+ ( x,t,\lambda) \tilde{ \phi}_- ( x,t,\lambda)
\end{split}
\end{equation}
where by "hat" we denote matrix inverse.

The next step of \cite{Beals-Coifman} was to prove that one can construct piece-wise FAS $\tilde{\chi}_\nu(x,t,\lambda)$
which allows analytic extension in a certain sector $\Omega_\nu$ in the complex $\lambda$-plane.
These results were generalized to any simple Lie algebra in \cite{VYa,ContM}.
The result is  that sector $\Omega_\nu$ has as boundaries  the rays starting from the origin
$l_{\nu-1}$ and $l_\nu$, see the Figures below.
The rays $l_\nu$ are determined by the solution of the linear equations
\begin{equation}\label{eq:lnu}
\begin{split}
 \Im \lambda \alpha(\tilde{J}) = 0.
\end{split}
\end{equation}

In what follows we will outline the construction of FAS for the operator $\mathfrak{L}$ which is defined by
\begin{equation}\label{eq:Ltt}
\begin{split}
\mathfrak{ L} \equiv \i \frac{\partial \xi}{ \partial x } + \tilde{Q}(x,t)\xi (x,t,\lambda) - \lambda [\tilde{J}, \xi(x,t,\lambda)] =0.
\end{split}
\end{equation}
Obviously the fundamental solutions of $\tilde{L}$ and $\mathfrak{L}$ are related by
\begin{equation}
\label{eq:xipm}
\xi_\pm (x,t,\lambda) =\tilde{ \phi}_\pm (x,t,\lambda) \e^{\i\lambda \tilde{ J}x}.
\end{equation}

The Jost solutions $\xi_\pm (x,t,\lambda)$ must satisfy Volterra type integral equations
\begin{equation}
\label{eq:xipm'}
\begin{split}
\xi_+(x,t, \lambda) &= \openone + \i \int_{\infty}^{x} \d y \; \e^{-\i\lambda \tilde{ J}(x-y)} Q(y,t) \xi_+ (y,t,\lambda) \e^{\i\lambda \tilde{ J}(x-y)}\\
\xi_-(x,t, \lambda) &= \openone + \i \int_{\infty}^{x} \d y \; \e^{-\i\lambda \tilde{ J}(x-y)} Q(y,t) \xi_- (y,t,\lambda)  \e^{\i\lambda \tilde{ J}(x-y)}. \\
\end{split}
\end{equation}

Let us now formulate the basic properties of $\xi_\nu (x,t, \lambda) $ --- the FAS of $\mathfrak{L}$.

\begin{enumerate}
\item
The continuous spectrum of $\mathfrak{L}$ fills up the rays $l_\nu$.

 \item
 Due to the  $\mathbb{Z}_h$ symmetry the rays $l_\nu$ close angles equal to $\pi/2h$ or $\pi/h$ depending
 on the choice of the algebra.

\item
To each ray $l_\nu$ we associate a subset of roots $\delta_\nu$ of $\mathfrak{g}$ which satisfy the condition (\ref{eq:lnu}).
Thus to each ray $l_\nu$ we associate a subalgebra $\mathfrak{g}_\nu \subset \mathfrak{g}$ generated by
$E_\alpha, E_{-\alpha}, H_\alpha$ for $\alpha\in \delta_\nu$.

\item
In each of the sectors  $\Omega_\nu $ we  can calculate the limits for $x\to\pm\infty$ along the lines $l_\nu$; more specifically
\begin{equation}
\label{eq:lim1}
\begin{split}
\lim_{x\to -\infty} \e^{-\i\lambda Jx} \xi_\nu (x,t,\lambda ) \e^{\i\lambda Jx} &= S_\nu^+ (t,\lambda)\\
\lim_{x\to \infty} \e^{-\i\lambda Jx} \xi_\nu (x,t,\lambda ) \e^{\i\lambda Jx}& = T_\nu^-(t,\lambda) D_\nu^+(\lambda)
\end{split}
\quad \forall \lambda \in l_\nu \e^{+ \i 0}
\end{equation}
and
\begin{equation}
\label{eq:lim2}
\begin{split}
\lim_{x\to -\infty} \e^{-\i\lambda Jx} \xi_\nu (x,t,\lambda)  \e^{\i\lambda Jx} &= S_{\nu+1}^- (t,\lambda)\\
\lim_{x\to -\infty} \e^{-\i\lambda Jx} \xi_\nu (x,t,\lambda ) \e^{\i\lambda Jx} &= T_{\nu+1}^+(t,\lambda) D^-_{\nu+1}(\lambda)
\end{split}
\quad \forall \lambda \in l_{\nu+1} \e^{- \i 0}
\end{equation}
where $S_\nu^\pm$,  $T_\nu^\pm$ and  $D_\nu^\pm$ are given by
\begin{equation}\label{eq:CTDpm}
\begin{aligned}
S_\nu^\pm (\lambda,t)   &= \exp \left( \sum_{\alpha\in\delta_\nu} s^\pm_\alpha(\lambda,t) E_{\pm \alpha}\right) \quad
D_\nu^\pm (\lambda)    = \exp \left( \sum_{\alpha\in\delta_\nu} d^\pm_{\nu, \alpha} H_{\alpha}\right) \\
T_\nu^\pm (\lambda,t)  &= \exp \left( \sum_{\alpha\in\delta_\nu} t^\pm_{\nu, \alpha}(\lambda,t) E_{\pm \alpha}\right) .
\end{aligned}\end{equation}
Obviously they take values in the subgroup $\mathcal{G}_\nu$
whose Lie algebra $\mathfrak{g}_\nu$ has as positive roots the subset of roots related to $l_\nu$, see the Table \ref{tab:1}.

\item
The time-dependence of  $S_\nu^\pm$,  $T_\nu^\pm$ and  $D_\nu^\pm$ is determined by the $M$ operator as
\begin{equation}\label{eq:dSpmdt}
\begin{aligned}
&
\i \frac{\partial S_\nu^\pm}{ \partial t } - \lambda^3 [K, S_\nu^\pm (\lambda,t) ] = 0 \qquad \i \frac{\partial D_\nu^\pm}{ \partial t }  = 0 \\
&
\i \frac{\partial T_\nu^\pm}{ \partial t } - \lambda^3 [K, T_\nu^\pm (\lambda,t) ] = 0
\end{aligned}
\end{equation}
where $K$ determines the leading term of the $M$-operator.

\item
The asymptotics $S_0^\pm$,  $T_0^\pm$ and $D_0^\pm$ and $S_1^\pm$,  $T_1^\pm$ and $D_1^\pm$ can be considered as
  independent. All the others are obtained from  them by the $\mathbb{Z}_h$ symmetry
  \begin{equation}\label{eq:Snu2}
  \begin{aligned}
  S_{2\nu}^\pm (\lambda) &= C^\nu (S_0^\pm(\lambda \omega^\nu), &\quad   S_{2\nu+1}^\pm (\lambda) &= C^\nu (S_1^\pm(\lambda \omega^\nu) \\
  T_{2\nu}^\pm (\lambda) &= C^\nu (T_0^\pm(\lambda \omega^\nu), &\quad   T_{2\nu+1}^\pm (\lambda) &= C^\nu (T_1^\pm(\lambda \omega^\nu) \\
  D_{2\nu}^\pm (\lambda) &= C^\nu (D_0^\pm(\lambda \omega^\nu), &\quad   D_{2\nu+1}^\pm (\lambda) &= C^\nu (D_1^\pm(\lambda \omega^\nu).
  \end{aligned}
  \end{equation}
\end{enumerate}

As a consequence of the above properties we prove the following lemma, which generalizes the results of Zakharov and Shabat
\cite{ZaSha} for this type of algebras.

\begin{lemma}\label{lem:1}
\begin{enumerate}
\item
The FAS $\xi_{\nu}(x,t,\lambda)$ of $\mathfrak{L}$ are solutions of the RHP
\begin{equation}\label{eq:rhp1}\begin{split}
\xi_{\nu+1}(x,t,\lambda) &=\xi_{\nu}(x,t,\lambda)G_{\nu}(x,\lambda), \\
G_{\nu}(x,\lambda) &=\e^{-\i\lambda Jx}\hat{S}^{-}_{\nu+1}S^{+}_{\nu+1}\e^{\i\lambda Jx}
\end{split}\end{equation}
which allows canonical normalization
\begin{equation}\label{eq:xi-can}
\lim_{\lambda\to\infty}\xi_{\nu}(x,t,\lambda)=\openone.
\end{equation}
\item
The corresponding potential $\tilde{Q}(x,t)$ is reconstructed from $ \xi_\nu(x,t,\lambda)$ by
\begin{equation}
\label{eq:Qxt}
\begin{split}
\tilde{ Q}(x,t) = \lim_{\lambda\to\infty} \lambda \left( \tilde{ J} - \xi_\nu(x,t,\lambda) \tilde{ J} \hat{\xi}_\nu(x,t,\lambda) \right)
\end{split}
\end{equation}
where $\xi_\nu(x,t,\lambda) $ is the unique regular solution of the RHP (\ref{eq:rhp1}), \cite{NMPZ,ContM,GeYa*94}.
\end{enumerate}
\end{lemma}

\begin{proof}
1) follows easily from eqs. (\ref{eq:lim1}), (\ref{eq:lim2}) and from the fact, that the FAS is determined uniquely by
its asymptotic for $x\to\pm\infty$.

2) follows from the fact that $\xi_\nu(x,t,\lambda)$ is a fundamental solution of $\mathfrak{L}$. Multiply eq. (\ref{eq:Ltt}) by
$\hat{\xi}_\nu(x,t,\lambda)$ on the right, take the limit $\lambda\to\infty$ and use the canonical normalization (\ref{eq:xi-can}).

\end{proof}

We will formulate the specific properties  for the 3 algebras independently.

\subsection{$B^{(1)}_{2}$}

Here $h=4$ and $\tilde{J}= \sqrt{2}\, \diag(1, \i, 0, -\i, -1)$. The rays $l_\nu$ are defined by $l_\nu : \arg \lambda = \nu \pi/4$; thus they
close angles $\pi/4$. The sectors
$\Omega_\nu$, $\nu=0,\dots,7$ are shown on Figure \ref{fig:1}, left panel. The set of roots $\delta_\nu$ related to each $l_\nu$ are given in
Table \ref{tab:1}.

\begin{figure}
\includegraphics[width= 60mm]{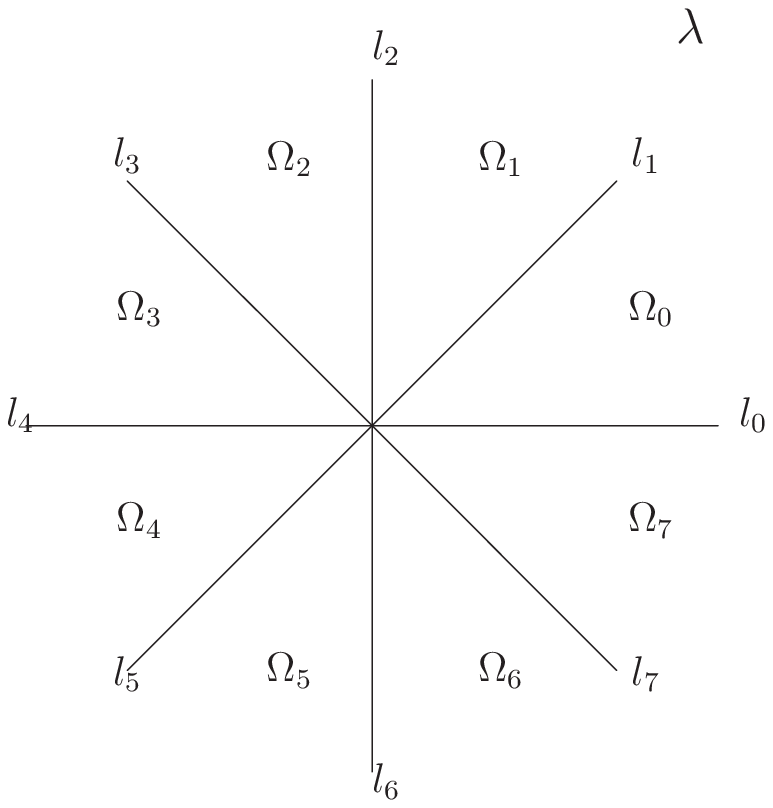}\qquad   \includegraphics[ width= 60mm]{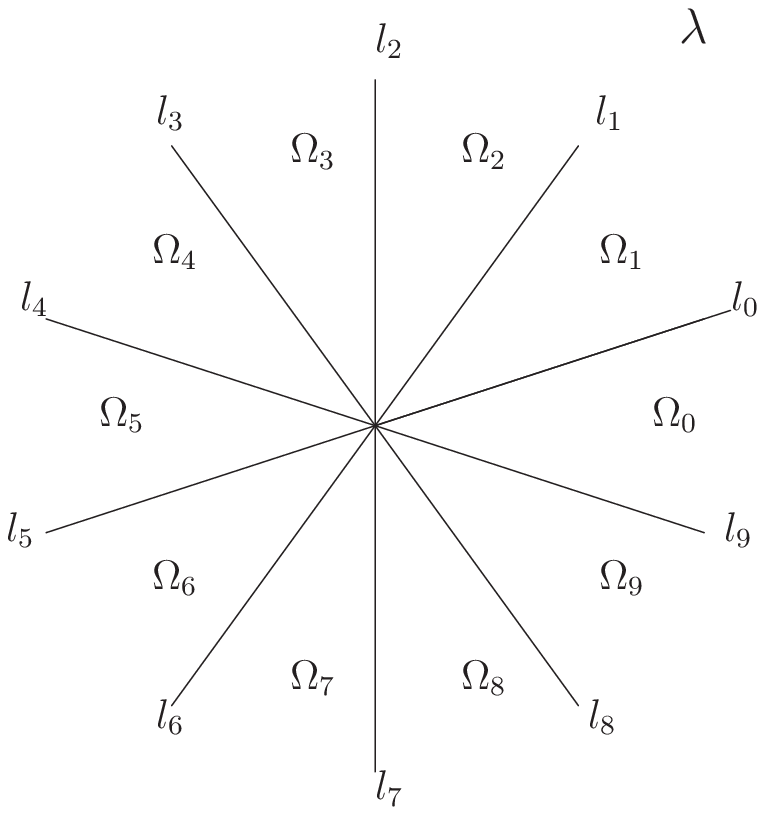}\\
\caption{
The continuous spectrum and the analyticity sectors of the FAS for the Lax operators: the case of $B^{(1)}_{2}$ ---   left panel;
the case of $A^{(2)}_{4}$ --- right panel; }
\label{fig:1}
\end{figure}

\begin{table}
  \centering
  \caption{Subsets of positive roots related to the lines $l_\nu\cup l_{\nu+4}$, $\nu=0,\dots, 3$ for the algebra $B_2^{(1)}$}\label{tab:1}
\begin{tabular}{|c|c|c|c|c|c|}
  \hline
 $ l_0\cup l_4 $ & $ l_1\cup l_5 $ & $ l_2\cup l_6 $ & $ l_3\cup l_7 $ \\
 $ e_1 $ & $ e_1-e_2 $ & $ e_2 $ & $ e_1+e_2 $ \\
  \hline
\end{tabular}
\end{table}

\subsection{$A^{(2)}_{4}$}

Similarly for $A^{(2)}_{4}$ we have  $h=10$ and $\tilde{J}= \diag(\omega,\omega^{3},-1,\omega^{-3},\omega^{-1})$ with $\omega = \exp (2\pi \i/10)$.
The rays $l_\nu$ are defined by $l_\nu : \arg \lambda = (2\nu +1)\pi/10$,  $\nu=0,\dots,9$; thus they
close angles $\pi/5$. The sectors
$\Omega_\nu$, $\nu=0,\dots,9$ are shown on Figure \ref{fig:1}, right panel.
The set of roots $\delta_\nu$ related to each $l_\nu$ are given in
Table \ref{tab:2}.

\begin{table}
  \centering
  \caption{Subsets of positive roots related to the lines $l_\nu\cup l_{\nu+5}$ $\nu=0,\dots, 4$ for the algebra $A_4^{(2)}$}\label{tab:2}
\begin{tabular}{|c|c|c|c|c|c|}
  \hline
 $ l_0\cup l_5 $ & $ l_1\cup l_6 $ & $ l_2\cup l_7 $ & $ l_3\cup l_8 $ & $ l_4\cup l_9 $ \\
 $ e_{1}-e_{2},e_{3}-e_{5} $ & $ e_{2}-e_{5},e_{3}-e_{4} $ & $ e_{1}-e_{5},e_{2}-e_{4} $ & $ e_{1}-e_{4},e_{2}-e_{3} $ & $ e_{1}-e_{3},e_{4}-e_{5} $ \\
  \hline
\end{tabular}
\end{table}

\subsection{$A^{(2)}_{5}$}

For $A^{(2)}_{5}$ we also have  $h=10$ but now $\tilde{J}= \diag(\omega,\omega^{3},-1,\omega^{-3},\omega^{-1})$ with $\omega = \exp (2\pi \i/10)$.
The rays $l_\nu$ now are defined by $l_\nu : \arg \lambda = \nu \pi/10$,  $\nu=0,\dots,19$; thus they
close angles $\pi/10$. The sectors
$\Omega_\nu$, $\nu=0,\dots, 19$ are shown on Figure \ref{fig:2}. The set of roots $\delta_\nu$ related to each $l_\nu$ are given in
Table \ref{tab:3}.

\begin{figure}
 \includegraphics[width= 70mm]{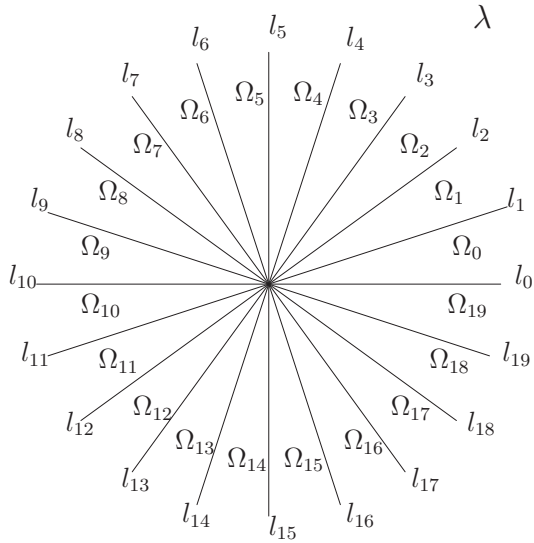}\\
 \caption{The continuous spectrum and the analyticity sectors of the FAS for the Lax operators for the case of $A^{(2)}_{5}$;  }\label{fig:2}
\end{figure}

\begin{table}
  \centering
  \caption{Subsets of positive roots related to the lines $l_\nu\cup l_{\nu+10}$, $\nu=0,\dots, 9$ for the algebra $A_5^{(2)}$}\label{tab:3}
\begin{tabular}{|c|c|c|c|c|c|}
  \hline
 $ l_{0}\cup l_{10} $ & $ l_{1}\cup l_{11} $ & $ l_{2}\cup l_{12} $ & $ l_{3}\cup l_{13} $ & $ l_{4}\cup l_{14} $ \\
 $ e_{1}-e_{6} $ & $ e_{1}-e_{3}, e_{4}-e_{5} $ & $ e_{3}-e_{6} $ & $ e_{1}-e_{2}, e_{3}-e_{5} $ & $ e_{5}-e_{6} $ \\
  \hline
  $ l_{5}\cup l_{15} $ & $ l_{6}\cup l_{16} $ & $ l_{7}\cup l_{17} $ & $ l_{8}\cup l_{18} $ & $ l_{9}\cup l_{19} $ \\
  $ e_{2}-e_{5},e_{3}-e_{4} $ & $ e_{2}-e_{6} $ & $ e_{1}-e_{5},e_{2}-e_{4} $ & $ e_{4}-e_{6} $ & $ e_{1}-e_{4},e_{2}-e_{3} $ \\
  \hline
\end{tabular}
\end{table}

We end this Section by the following lemma
\begin{lemma}\label{lem:2}
Each of the subalgebras $\mathfrak{g}_\nu$ related to the ray $l_\nu $ is a direct sum of $sl(2)$ subalgebras.
\end{lemma}

\begin{proof}
Let us prove our lemma for the algebra $A_{5}^{(2)}$. First we consider the subalgebras $\mathfrak{g}_0$ and  $\mathfrak{g}_1$
related to the rays $l_0$ and $l_1$. From Table V we find that the algebra  $\mathfrak{g}_0$ is generated by $E_\alpha$,
$E_{-\alpha}$ and $H_\alpha$, where $\alpha$ takes the values $e_1-e_4$, $e_2-e_3$ and $e_5-e_6$. These three roots are
mutually orthogonal, which means that  $\mathfrak{g}_0\equiv sl(2)\oplus sl(2)\oplus sl(2)$. Similarly, the algebra
$\mathfrak{g}_1$ is generated by $E_\beta$, $E_{-\beta}$ and $H_\beta$, where $\beta$ takes the values $e_1-e_3$ and
$e_4-e_6$, which are orthogonal to each other.  Therefore $\mathfrak{g}_1\equiv sl(2)\oplus sl(2)$.
Next we use the $\mathbb{Z}_{r+1}$ symmetry, which in particular means that the set of the roots $\delta_\nu$ related to
the ray $l_\nu$ by the Coxeter transformation $C$ as follows
\begin{equation}\label{eq:Cox}\begin{split}
\delta_{2\nu} = C^\nu (\delta_0), \qquad \delta_{2\nu+1} = C^\nu (\delta_1).
\end{split}\end{equation}
It remains to use the fact that the Coxeter transformation is an orthogonal transformation of the space of roots, so it
obviously preserves the angles between any two roots.

The other cases are proved analogously.
\end{proof}

\section{Discussion and Conclusion}

We have derived several systems of equations which are related to the affine Kac-Moody algebras $B^{(1)}_{2}$, $A^{(2)}_{4}$ and
$A^{(2)}_{5}$ respectively. They admit a Lax representation and can be solved using Inverse scattering method. We also outlined
the spectral properties of their Lax operators and formulated the corresponding RHP. This can be used to derive their
soliton solutions via the dressing Zakharov-Shabat method.  Lemma \ref{lem:2} can be used to prove the complete integrability of these mKdV equations. One can also develop the spectral theory of the relevant recursion operators following the ideas of
\cite{Konop,ContM,VG-Ya-13,GuKV} which can be used as a ground for uniform deriving of all fundamental properties of the NLEE.

\section*{Acknowledgments}
The work is partially supported by the ICTP -- SEENET-MTP project PRJ-09.
One of us (VSG) is grateful to professor A.S. Sorin for useful discussions during his visit to JINR, Dubna, Russia under project 01-3-1116-2014/2018.

\end{document}

%% file: somedef.tex
\def\diag{
{\rm diag}}

\newcommand{\e}{\mathrm{e}}
\renewcommand{\i}{\mathrm{i}}
\renewcommand{\d}{\mathrm{d}}

\theoremstyle{plain}
\newtheorem{theorem}{Theorem}

\newtheorem{lemma}[theorem]{Lemma}

\newenvironment{proof}[1][{}]{\par\noindent\textbf{Proof{#1}: }}{\hspace*{\fill}$\blacksquare$\smallskip\noindent\par}